\documentclass[9pt,shortpaper,twoside,web]{ieeecolor}
\usepackage{generic}

\usepackage{amsmath,nccmath} % assumes amsmath package installed
\usepackage{mathtools}
\usepackage{relsize}
\usepackage{tikz}
\graphicspath{{./Figures/}} 

\usepackage[margin=0.8cm]{caption}
\newcommand{\norm}[1]{\left\lVert#1\right\rVert}
\usepackage{pgfplots}
\pgfplotsset{compat=newest}
\usepackage{tabularx}
\pgfplotsset{plot coordinates/math parser=false}
\usepackage{algorithm}
\usepackage{soul}
\usepackage{dsfont}
\usepackage{tabto}

\usepackage{color}

\newtheorem{assumption}{Assumption}

\newtheorem{remark}{Remark}

\newtheorem{proposition}{Proposition}

\newtheorem{definition}{Definition}

\usepackage{lineno,amssymb,subcaption,algpseudocode}

\usepackage{algorithm}
\usepackage{algpseudocode}

\usepackage{booktabs}
\newcommand\scalemath[2]{\scalebox{#1}{\mbox{\ensuremath{\displaystyle #2}}}}

\newcommand{\A}{A}

\newcommand{\R}{\mathbb{R}}

\newcommand{\1}{\mathbf{1}}
\newcommand{\0}{\mathbf{0}}
\newcommand{\I}{\mathbf{I}}
\usepackage{accents}

\newcommand{\seq}[2]{\mathbb{I}_{#1}^{#2}}

\usepackage{cuted}
\usepackage{amsmath,amssymb,amsfonts}
\usepackage{graphicx}
\usepackage{paralist}

\usepackage{verbatim}
\usepackage[normalem]{ulem}

\usepackage{bm}
\usepackage{tikz}
\usetikzlibrary{calc,patterns,decorations.pathmorphing,decorations.markings}
\usepackage{amsmath} % assumes amsmath package installed
\usepackage{amssymb}  % assumes amsmath package installed
\usepackage{mathtools}
\usepackage{dsfont}
\usepackage{bm}
\usepackage{xcolor}
\definecolor{blue_set}{RGB}{204,229.5,255}
\definecolor{pink_set}{RGB}{255,204,229.5}
\definecolor{grey_set}{RGB}{153,153,153}
\definecolor{green_set}{RGB}{102,204,102}
\definecolor{cyan_set}{RGB}{69,243,248}
\definecolor{yellow_set}{RGB}{229.5000,229.5000,114.7500}
\definecolor{red_set}{RGB}{231,172,116}
\definecolor{red_border}{RGB}{229.5,114.75,114.75}
\definecolor{terminal_set}{RGB}{102,204,255}
\definecolor{mRPI_set}{RGB}{102,255,102}
\usepackage{filecontents}
\usetikzlibrary{positioning}
\usetikzlibrary{calc}
\usetikzlibrary{shapes,snakes}
\newcommand\munderbar[1]{\underaccent{\bar}{#1}}

\newcommand{\smallmat}[1]{\left[ \begin{smallmatrix}#1 \end{smallmatrix} \right]}

\newcommand{\MM}[1]{{\color{red}{#1}}}

\usepackage{url}

\definecolor{wheat}{rgb}{0.96,0.87,0.70}
\definecolor{mario}{rgb}{0.8,0.8,1}
\definecolor{SamComm}{rgb}{0.9,0.9,0.1}
\definecolor{ao}{rgb}{0.0, 0.5, 0.0}

\nocite{Gupta2023} %Unique first citation always

\title{\LARGE \bf
	Parameter-Dependent Robust Control Invariant Sets for LPV Systems with Bounded Parameter-Variation Rate
}

\author{Sampath Kumar Mulagaleti, Manas Mejari, Alberto Bemporad
	\thanks{S.K. Mulagaleti is with University of Trento, Via Sommarive, 9, 38123 Povo, Trento, Italy (\texttt{sampath.mulagaleti@unitn.it}), M. Mejari is with IDSIA Dalle Molle Institute for Artificial Intelligence, USI-SUPSI,  Via la Santa 1, CH-6962 Lugano-Viganello, Switzerland (\texttt{manas.mejari@supsi.ch}), and A. Bemporad is with IMT School for Advanced Studies Lucca, Piazza S.Francesco 19, 55100 Lucca, Italy (\texttt{alberto.bemporad@imtlucca.it}).
 }%
}

\begin{document}

	\maketitle
	\thispagestyle{empty}
	\pagestyle{empty}

	%%%%%%%%%%%%%%%%%%%%%%%%%%%%%%%%%%%%%%%%%%%%%%%%%%%%%%%%%%%%%%%%%%%%%%%%%%%%%%%%
	\begin{abstract}
 Real-time measurements of the scheduling parameter of \emph{linear parameter-varying} (LPV) systems enable the synthesis of \emph{robust control invariant} (RCI) sets and parameter dependent controllers inducing invariance. We present a method to synthesize \emph{parameter-dependent robust control invariant} (PD-RCI) sets for LPV systems with bounded parameter variation, in which invariance is induced using PD-vertex control laws.
The PD-RCI sets are parameterized as \emph{configuration-constrained} polytopes that admit a joint
parameterization of their facets and vertices. The proposed sets and associated control laws are computed by solving a single semidefinite programing (SDP) problem. Through numerical examples, we demonstrate that the proposed method outperforms state-of-the-art methods for synthesizing PD-RCI sets, both with respect to conservativeness and computational load.  
	\end{abstract}
\section{Introduction}	
Robust control invariant (RCI) sets are subsets of the state-space in which a dynamical system can be forced to evolve indefinitely in the presence of arbitrary but bounded disturbances. Such sets form the basis for the analysis and design of control schemes, since they define the regions in which the system can be forced to operate~\cite{Blanchini2015,sr10}. Hence, the development of methods to characterize and compute RCI sets is an active research area.
These approaches can broadly be divided into two categories. The first category is related to recursive approaches, in which the limit set of finite-time controllable sets is computed, see, e.g.,~\cite{sr10,Bertsekas72,Kolmanovsky1998,Rungger2017,anevlavis2022controlled}. In the second category, RCI sets of an a priori selected representation are computed by enforcing invariance using an a priori chosen controller parameterization, see, e.g.,~\cite{Gupta2019,Liu2015}. This paper presents a method to compute RCI sets for LPV systems in the latter setting.

A common approach to compute RCI sets for LPV systems involves considering the scheduling parameter as an arbitrary disturbance~\cite{Hanema2020,sm05,ag19b,Nguyen15}. However, this parameter is typically measured before computing the control input, and exploiting this information can help synthesize RCI sets with reduced conservativeness. This rationale is the basis for numerous parameter-dependent (PD) control design schemes~\cite{garone2009_thesis}. While the RCI sets corresponding to these PD-controllers can be synthesized a posteriori, to the best of our knowledge, the technique proposed in~\cite{Gupta2023} is the only one that provides the simultaneous synthesis of PD-control laws and their asscociated PD-RCI sets. 
The PD-RCI sets serve to identify regions within the state-space from which invariance can be achieved using the corresponding PD-control law, exploiting
the available information on the scheduling parameter.

In this paper, we present a new approach to simultaneously synthesize PD-RCI sets and corresponding PD-control laws using a single semidefinite programing problem (SDP). The main difference with respect to~\cite{Gupta2023} stems from the representation of the  PD-RCI sets, and the parameterization of the corresponding invariance-inducing control laws. We represent PD-RCI sets as polytopes having fixed orientation and varying offsets that depend affinely on the scheduling parameter, and induce invariance in these sets using a PD-vertex control law~\cite{Gutman1986}. Since every linear feedback law can be transformed into a vertex control law, our parameterization is inherently less conservative than that proposed in~\cite{Gupta2023}. Furthermore, we enforce \textit{configuration constraints}~\cite{villanueva2022configurationconstrained} on these polytopes, that enables a convex formulation of the PD-RCI set synthesis problem, rather than a much more computationally expensive nonlinear matrix problem developed in~\cite{Gupta2023}. Finally, unlike in~\cite{Gupta2023}, we seamlessly incorporate information of the rate of parameter variation into the PD-RCI computation problem. It is well known that taking into account the bounds on rate of variation 
can reduce conservativeness in the control design procedure and consequently in the computation of the RCI sets~\cite{Oliveira2008,Hanema2020,Morato2020}. Through numerical examples, we demonstrate that our approach computes PD-RCI sets with reduced conservativeness at a much lower computational expense compared to the approach of~\cite{Gupta2023}.

 This paper is organized as follows. In Section~\ref{sec:prob_defn}, the concept of PD-RCI sets for LPV systems is recalled, and the problem statement formulated. Then, in Section~\ref{sec:computation_PDRCI}, a semidefinite programming problem to compute PD-RCI sets is derived. Finally in Section~\ref{sec:examples}, numerical examples to support the efficacy of the approach, and a comparison to other state-of-the-art methods are presented.

\textit{Notation:} Given a matrix $L\in \R^{m \times n}$, we denote by $L\mathcal{X}$ the image $\{y\in\R^{m}: y=Lx, x \in \mathcal{X}\}$ of a set $\mathcal{X} \subseteq \R^n$ under the linear transformation induced by $L$. We denote the $i$-th row of matrix $L$ by $L_i$. The symbols $\1^{n \times m}$ and $\0^{n \times m}$ denote all-ones and all-zeros matrices in $\R^{n \times m}$ respectively, and $\I^n$ denotes the identity matrix of order $n$. We ignore the superscript if sizes are clear from the context. The set $\mathbb{I}_m^n:=\{m,\ldots,n\}$ denotes the set of natural numbers between two integers $m$ and $n$, $m\leq n$. Given compatible matrices $A$ and $B$, $A \otimes B$ denotes their Kronecker product. The symbol $\R^{n \times m}_+$ denotes the set of all matrices in $\R^{n \times m}$ with nonnegative elements.
		The Minkowski set addition is defined as $\mathcal{X} \oplus \mathcal{Y}:=\{x+y:x\in\mathcal{X},y\in\mathcal{Y}\}$, and set subtraction as $\mathcal{X} \ominus \mathcal{Y}:=\{x:\{x\} \oplus \mathcal{Y} \subseteq \mathcal{X}\}$. Given points $\{x_i, i \in \mathbb{I}_1^N\}$, $\mathrm{CH}(x_i, i \in \mathbb{I}_1^N)$ denotes their convex-hull. Symmetric block matrices are denoted by $*$. A $p$-norm ball  is denoted by $\mathcal{B}_{p}^n := \{x \in \mathbb{R}^n: \left \| x \right \|_p \leq 1  \} $. 
  \begin{proposition}[Strong duality~\cite{Schneider2013}]\label{prop:strong_duality}
     Given $a \in \R^n$, $b \in \R$, $M \in \R^{m \times n}$ and $q \in \R^m$, the inequality $a^{\top} x \leq b$ holds for all $x$ such that $Mx \leq q$ if and only if there exists some $\Lambda \in \R^{1 \times m}_+$ satisfying $\Lambda q \leq b$ and $\Lambda M=a^{\top}$.    $\hfill\square$
  \end{proposition}
	\section{Problem definition}
 \label{sec:prob_defn}
 Consider the discrete-time LPV system with dynamics
\begin{align}
\label{eq:LPV_system}
x^+=A(p)x+B(p)u+w,
\end{align}
where $x \in \mathbb{R}^{n}$, $u\in \R^{m}$, $w \in \R^{n}$, and $p \in \R^s$ represent the state, input, additive disturbance, and scheduling parameter, respectively, and $x^+$ is the successor state. The matrices $A(p)$ and $B(p)$ depend linearly on the parameter $p$ as
\begin{align}
\label{eq:P_dependence}
A(p):=\underset{j=1}{\overset{s}{\scalemath{1}{\sum}}} p_j A^j, && B(p):=\underset{j=1}{\overset{s}{\scalemath{1}{\sum}}} p_j B^j,
\end{align}
where $A^j$ and $B^j$ are matrices of appropriate dimensions. 
    The system is subject to state constraints $x \in \mathcal{X}$, input constraints $u \in \mathcal{U}$, the additive disturbance $w \in \mathcal{W}$, and the scheduling parameter satisfies $p \in \mathcal{P}$.  
   Moreover, the parameter variation is assumed to be bounded in a given set $\mathcal{R}$, thus, for any $p \in \mathcal{P}$, the successive parameter $p^+$ is bounded by a set $\mathbb{P}(p)$ as
  \begin{align}
    \label{eq:bounded_param_variation}
    p^+ \in \mathbb{P}(p):=\left(\{p\} \oplus \mathcal{R}\right) \cap \mathcal{P}.
    \end{align}
     \begin{assumption}
        \label{ass:zero_in_deltaP}
        $\0^s \in \mathcal{R}$. $\hfill\square$
    \end{assumption}
    This assumption guarantees that $\mathbb{P}(p) \neq \emptyset$ for all $p \in \mathcal{P}$.
\vspace{5pt} \\
\indent
\textit{\textbf{Definition:}} \textit{A set $\mathcal{S}(p) \subseteq \R^{n}$ is a parameter-dependent robust control invariant (PD-RCI) set for System~\eqref{eq:LPV_system} if and only if}
%with constraints $x \in \mathcal{X}$, $u \in \mathcal{U}$, $w \in \mathcal{W}$, and $p \in \mathcal{P}$ if and only if% 
    \begin{align}
    \label{eq:LPV_conditions:1}
   & \quad \mathcal{S}(p) \subseteq \mathcal{X}, \ \forall \ p  \in \mathcal{P},  \\ \nonumber \vspace{-20pt}  \\
     \label{eq:LPV_conditions:2}
   &  \begin{cases}
    \forall \ (p,x) \in \mathcal{P} \times \mathcal{S}(p) , \ \forall \ (p^+,w) \in \mathbb{P}(p) \times \mathcal{W} , \  \vspace{3pt} \\
  \exists  \ 
 u=u(x,p) \in \mathcal{U} \ : \ \A(p)x+B(p) u +w \in \mathcal{S}(p^+). 
 \end{cases}
    \end{align} 
    $\hfill\square$
\\
Note that the inclusion in~\eqref{eq:LPV_conditions:2} is equivalent to
    \begin{align}
    \label{eq:LPV_conditions:3}
        \begin{matrix}
\{ A(p)x+B(p) u \} \oplus \mathcal{W} \subseteq \underset{p^+ \in \mathbb{P}(p)}{\bigcap} \mathcal{S}(p^+). \end{matrix}
    \end{align}
This definition of a PD-RCI set draws inspiration from~\cite{Gupta2023}. 
In a conventional setting, an RCI set characterizes a set of states in which a system can be forced to belong \textit{independently of the current parameter $p(t)$}.
In contrast, the PD-RCI set explicitly accounts for the current parameter value, resulting in an enlarged set of states from which invariance can be achieved. Unlike in~\cite{Gupta2023}, our definition also \textit{accounts for bounded parameter variations} $\mathcal{R}$, thus reducing conservativeness further. If $\mathcal{S}(p)$ verifies inclusions~\eqref{eq:LPV_conditions:1}-\eqref{eq:LPV_conditions:2}, then given any initial state-parameter pair $(x(0),p(0))$ with $p(0) \in \mathcal{P}$ and $x(0) \in \mathcal{S}(p(0))$, 
there exist inputs $u(t) \in \mathcal{U}$ enforcing $x(t) \in \mathcal{S}(p(t))$ for all future disturbances $w(t) \in \mathcal{W}$ and parameters $p(t) \in \mathcal{P}$ satisfying~\eqref{eq:bounded_param_variation} for all $t\geq 0$. 

We present an illustration of PD-RCI sets in Figure~\ref{fig:PD_RCI_illustration}. The left figure illustrates the parameter space, where $p^0$ is the current parameter value, gray region is the set $\mathcal{P}$, and thick black line is the set $\{p^0\} \oplus \mathcal{R}$. Then, we have $\mathbb{P}(p^0)=\mathrm{CV}\{p^1,p^2\}$. The right figure illustrates the state space. As per the condition in~\eqref{eq:LPV_conditions:3}, any $x \in \mathcal{S}(p^0)$ can be driven into the set $\cap_{p \in \mathbb{P}(p^0)} \mathcal{S}(p^j)$. Additionally, the set
\begin{align}
\label{eq:tilde_S}
\tilde{\mathcal{S}}:=\bigcup_{ p \in \mathcal{P}} \mathcal{S}(p)
\end{align}
is plotted. By definition, $\mathcal{S}(p) \subseteq \tilde{\mathcal{S}}$, and $\mathcal{S}(p)$ represents states that can be rendered invariant for a given parameter $p$.
\begin{figure}[t]
\centering
\vspace*{0.1cm}
\hspace*{0.cm}
{
\resizebox{0.41\textwidth}{!}
{
\begin{tikzpicture}
\begin{scope}[xshift=0.cm]
\node[draw=none,fill=none](tanks_fig) {\includegraphics[trim=0 0 0 0,clip,scale=3]{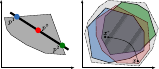}};
\end{scope}
\end{tikzpicture}}
}
\captionsetup{width=1\linewidth}
\caption{Illustration of PD-RCI sets. (\textit{Left: Parameter space}) The gray set is $\mathcal{P}$, the current parameter is $p^0$, and the thick black line is $\{p^0\} \oplus \mathcal{R}$, such that $\mathbb{P}(p^0)=\mathrm{CV}\{p^1,p^2\}$. (\textit{Right: State space}) Red, blue and green sets are $\mathcal{S}(p^0)$, $\mathcal{S}(p^1)$ and $\mathcal{S}(p^2)$ respectively. Hatched region is $\cap_{ p \in \mathbb{P}(p^0)} \mathcal{S}(p)$. Inclusion \eqref{eq:LPV_conditions:3} implies any $x \in \mathcal{S}(p^0)$ can be driven into the hatched region. The gray set with dot-dashed outline is $\tilde{\mathcal{S}}$ defined in~\eqref{eq:tilde_S}. This set includes $\mathcal{S}(p)$ for all $p \in \mathcal{P}$.}
\label{fig:PD_RCI_illustration}
%\vspace*{-0.7cm}
\end{figure}

\subsubsection*{Polytopic sets}
We restrict our discussion to the polytopic setting for brevity. We assume that System~\eqref{eq:LPV_system} is subject to the following constraints%
    \begin{subequations}
    \label{eq:polytope_definitons}
    \begin{align}
       & \mathcal{X}:=\{x: H^x x \leq h^x\},  &\mathcal{U}:=\{u : H^u u \leq h^u\}, \\
        & \mathcal{W}:=\{w : H^w w \leq h^w\}, & \mathcal{P}:=\{p : H^p p \leq h^p\}, \\
         & \mathcal{R}:=\{  p \in \R^s : H^{\delta}  p \leq h^{\delta}\},
    \end{align}
    \end{subequations}
    with $\scalemath{0.99}{h^x \in \R^{m_x}, \ h^u \in \R^{m_u}, \ h^w \in \R^{m_w}, \ h^p \in \R^{m_p}, \ h^{\delta} \in \R^{m_{\delta}}}$. We remark that the methodology in the sequel can be adapted to more general convex representations. In the following result, we characterize polytopic PD-RCI sets in vertex representation.
    \begin{proposition}
    \label{prop:PD_RCI_basic_cond}
    For any parameter $p \in \mathcal{P}$, suppose that there exists a parameterized polytope $\mathcal{S}(p)=\mathrm{CH}\{x^i(p),i \in \mathbb{I}_1^{v(p)}\} \subseteq \mathcal{X}$, and for each vertex $x^i(p)$, there exists an input $u^i(p) \in \mathcal{U}$ such that
    \begin{align}
    \label{eq:RCI_vert_condition}
        \{A(p)x^i(p)+B(p)u^i(p)\} \oplus \mathcal{W} \subseteq  \underset{p^+ \in \mathbb{P}(p)}{\bigcap}\mathcal{S}(p^+).
    \end{align}
    Then $\mathcal{S}(p)$ is a PD-RCI set, i.e., it satisfies~\eqref{eq:LPV_conditions:1} and~\eqref{eq:LPV_conditions:3}.  $\hfill\square$
    \end{proposition}
    \begin{proof}
        Condition~\eqref{eq:LPV_conditions:1} holds by definition of $\mathcal{S}(p)$. Regarding Condition~\eqref{eq:LPV_conditions:3}, consider any $x \in \mathcal{S}(p)$ for some $p \in \mathcal{P}$. By convexity, there exist $\lambda \in \R_+^{v(p)}$ satisfying $\1^{\top} \lambda = 1$ such that $x = \underset{i=1}{\overset{v(p)}{\scalemath{1.2}{\Sigma}}} \lambda_i x^i(p)$.
        Applying the input $u= \underset{i=1}{\overset{v(p)}{\scalemath{1.2}{\Sigma}}} \lambda_i u^i(p) \in \mathcal{U}$ yields, for any $w \in \mathcal{W}$,
        \begin{align*}
        \hspace{-5pt}
            x^+& = A(p) \underset{i=1}{\overset{v(p)}{\scalemath{1.2}{\Sigma}}} \lambda_i x^i(p)+B(p) \underset{i=1}{\overset{v(p)}{\scalemath{1.2}{\Sigma}}} \lambda_i u^i(p) + w \\
&\hspace{-10pt}= \underset{i=1}{\overset{v(p)}{\scalemath{1.2}{\Sigma}}} \lambda_i(  A(p) x^i(p)+B(p)u^i(p)+w)  \in  \underset{p^+ \in \mathbb{P}(p)}{\bigcap} \mathcal{S}(p^+)  
        \end{align*}
        where the last inclusion follows from~\eqref{eq:RCI_vert_condition} and convexity of the set $\{x \in \mathcal{S}(p^+) : p^+ \in \mathbb{P}(p)\}$.
    Since $p \in \mathcal{P}$ and $x \in \mathcal{S}(p)$ were arbitrary,~\eqref{eq:LPV_conditions:3} is satisfied by $\mathcal{S}(p)$. \\
    \end{proof}

  \textit{\textbf{Problem statement: }}
    \textit{Given the LPV system~\eqref{eq:LPV_system} subject to the constraints in~\eqref{eq:polytope_definitons},
    compute a polytopic PD-RCI $\mathcal{S}(p)$ with vertex control inputs $\{u^i(p) \in \mathcal{U}, i \in \mathbb{I}_1^{v(p)}\}$ verifying inclusions~\eqref{eq:LPV_conditions:1} and~\eqref{eq:RCI_vert_condition}.} $\hfill\square$
    \vspace{5pt} \\ \indent
\textit{Parameter-dependent control law:} As described in Proposition~\ref{prop:PD_RCI_basic_cond}, invariance in PD-RCI sets verifying inclusion~\eqref{eq:RCI_vert_condition} is induced using a vertex control law~\cite{Gutman1986}. At time $t$, the control input is given by
\begin{align}
\label{eq:cc_control_1}
    u(t)=\sum_{i=1}^{v(p(t))} \lambda_iu^i(p(t)),
\end{align}
where $\lambda \in \R^{v(p(t))}$ solves the quadratic program
\begin{align}
\label{eq:QP_vertex_inputs}
    \min_{\lambda \geq \0} \ ||\lambda||_2^2 \ \text{s.t.} \ \sum_{i=1}^{v(p(t))} \lambda_i x^i(p(t))=x(t), \ \ \1^{\top}\lambda=1.
\end{align}
 In~\eqref{eq:cc_control_1}-\eqref{eq:QP_vertex_inputs}, $x^i(p)$ and $u^i(p)$ are parameter-dependent vertices and the corresponding control inputs respectively, the functional forms of which we present in the sequel. By construction, Problem~\eqref{eq:QP_vertex_inputs} is feasible if $x(t) \in \mathcal{S}(p(t))$. Moreover, if $x(0) \in \mathcal{S}(p(0))$, then Problem~\eqref{eq:QP_vertex_inputs} is feasible for all $t\geq0$. 
%The PD-RCI set $\mathcal{S}(p)$ we compute can be used to select a control input online in a way similar to standard vertex control law schemes~\cite{Gutman1986}. After measuring the current state $x(t)$ and parameter $p(t)$, if $x(t) \in \mathcal{S}(p(t))$, then the control input 
%\begin{align}
%\label{eq:cc_control_1}
%    u(t)=\sum_{i=1}^{v(p(t))} \lambda_iu^i(p(t)),
%\end{align}
%is applied,
%where $\lambda \in \R^{v(p(t))}$ solves the quadratic program
%\begin{align*}
%    \min_{\lambda \geq \0} \ ||\lambda||_2^2 \ \text{s.t.} \ \sum_{i=1}^{v(p(t))} \lambda_i x^i(p(t))=x(t), \ \ \1^{\top}\lambda=1.
%\end{align*}
    \begin{remark}
        The following results can be generalized to accommodate constraints $\mathcal{Z}:=\{y:H^y y \leq h^y\}$ on the parameter-dependent output $y=C(p)x+D(p)u$ via minor adaptations.   $\hfill\square$
    \end{remark}
    \begin{remark}
       If full state-parameter measurements are unavailable, PD-RCI sets can be computed for observer dynamics derived using a PD-observer synthesized via, e.g.~\cite{Heemels2010}. Further research into a concurrent synthesis approach is a subject of future study.   $\hfill\square$
    \end{remark}
    \section{Computation of PD-RCI sets}
    \label{sec:computation_PDRCI}
  We now present an approach to compute PD-RCI sets. We first introduce a parameterization of the sets in Section~\ref{sec:cc_polytopes}, for which we encode inclusions~\eqref{eq:LPV_conditions:1} and~\eqref{eq:RCI_vert_condition} in Section~\ref{sec:inclusions}. Then, we formulate a convex optimization problem to compute the largest feasible PD-RCI set in Section~\ref{sec:volume}. In Section~\ref{sec:selecting_C} we provide a discussion of an approach to select a set parameterization that guarantees feasibility of the optimization problem.
   % Then, we present convex encodings of the inclusions in~\eqref{eq:LPV_conditions:1} and~\eqref{eq:RCI_vert_condition} for the proposed parameterization, using which formulate an optimization problem to 
    \subsection{Configuration-constrained polytopes}\label{sec:cc_polytopes}
   The main challenge in encoding the condition in~\eqref{eq:RCI_vert_condition} arises from the fact that, in general, it is intractable to characterize the intersection set on the right-hand-side of the inclusion since the sets $\mathcal{S}(p^+)$ are in their vertex representation~\cite{Tiwary2008}. To tackle this issue, a general methodology can be derived by adopting results on parameterized polytopes~\cite{Loechner1997}. In this paper, we tackle this challenge by leveraging recent results from~\cite{villanueva2022configurationconstrained} regarding constraints on facet representations.
    To this end, we parameterize the set $\mathcal{S}(p)$ in hyperplane representation as
    \begin{align}
        \label{eq:PD_RCI_parameterization}
        \mathcal{S}(p) \leftarrow \mathcal{S}(p|y^0,Y):=\left\{x : Cx \leq y^0+Yp\right\},
    \end{align}
    where $C \in \R^{m_s \times n}$ is a user-given matrix whose rows encode the normal vectors to the facets of $\mathcal{S}(p|y^0,Y)$. The offsets of these facets are affinely dependent on $p$ as $y^0+Yp$, where $y^0 \in \R^{m_s}$ and $Y:=[y^1 \cdots y^s] \in \R^{m_s \times s}$.
    Over the offset vector, we enforce \textit{configuration constraints}~\cite{villanueva2022configurationconstrained}: For a given matrix $\mathbf{E}$, these constraints are represented by the cone $\mathbb{S}:=\{y:\mathbf{E} y \leq \0\}$. They fix the facial configuration of the parameteric polytope $\mathcal{S}(p|y^0,Y)$ such that
    \begin{align}
        \label{eq:config_constraint_relation}
        \hspace{-5pt}
       \scalemath{0.99}{ y^0+Yp \in \mathbb{S} \Rightarrow \mathcal{S}(p|y^0,Y)=\mathrm{CH}\left\{V^k(y^0+Yp), k \in \mathbb{I}_1^N\right\}}
    \end{align}
    for any parameter $p$,
    where the matrices $\{V^k \in \R^{n_x \times m_s}, k \in \mathbb{I}_1^N\}$ capture the linear maps from the offset vector to the vertices of $\mathcal{S}(p|y^0,Y)$. For details regarding the construction of matrices $\mathbf{E}$ and $V^k$, we refer the reader to Appendix~\ref{sec:appendix}.
    %
    %
    %i.e., as a polytope with fixed normal vectors given in $C$, and the right-hand-side that depends affinely on the parameter $p$.
    %We enforce \textit{configuration-constraints}~\cite{villanueva2022configurationconstrained} over $\mathcal{S}(p|y^0,Y)$,
    %We enforce  to be a  polytope,
    %that enable us to jointly parameterize the halfspace and vertex representations of $\mathcal{S}(p|y^0,Y)$ in terms of $(y^0,Y)$. Given a polytope $\mathcal{Y}(y):=\{x:Cx \leq y\}$, the configuration constraints over $y$ are described by the cone
    %\begin{align}
    %\label{eq:config_constraint_set}
    %    \mathbb{S}:=\{ y \ : \ \mathbf{E} y \leq \0\}.
    %\end{align}
    %Let matrices $\scalemath{0.95}{\{V^k \in \R^{n \times m_s}, k \in \mathbb{I}_1^N\}}$ define the vertex maps of $\mathcal{Y}(y)$,  
    %$$\mathcal{Y}(y)=\mathrm{CH}\{V^k y, k \in \mathbb{I}_1^N\}$$ for a given $y$. Then, the configuration constraints dictate that for a particular construction of matrices $\{V^k,k\in \mathbb{I}_1^N,\mathbf{E}\}$, it holds that
    %\begin{align}
    %\label{eq:config_constraint_relation}
    %     \forall \ y \in \mathbb{S}, && \mathcal{Y}(y)=\mathrm{CH}\{V^ky, k \in \mathbb{I}_1^N\}.
    %\end{align}
    %In the rest of this paper, 
    %For a user-specified matrix $C$ parameterizing $\mathcal{S}(p|y^0,Y)$, we assume to be given the matrices $\{V^k,k\in \mathbb{I}_1^N,\mathbf{E}\}$ satisfying~\eqref{eq:config_constraint_relation}. For further details regarding their constructions, we refer the reader to Appendix~\ref{sec:appendix}. 
    To exploit the result in~\eqref{eq:config_constraint_relation} for synthesizing PD-RCI sets, we make the following assumption. 
    \begin{assumption} 
 \label{ass:positive_param}
$p \geq \0^s$ for all $p \in \mathcal{P}$. $\hfill\square$
 \end{assumption}
 Under Assumption~\ref{ass:positive_param}, we have that if $y^j \in \mathbb{S}$ for all $j \in \mathbb{I}_0^s$, then $y^0 + Yp \in \mathbb{S}$ for any $p \in \mathcal{P}$ because it is a conic combination. Then, from~\eqref{eq:config_constraint_relation}, it follows that for any $p \in \mathcal{P}$, the polytope
\begin{align}
    \hspace{-10pt}\mathcal{S}(p|y^0,Y)=\mathrm{CH}\left\{x^k(p):=V^k (y^0 +Y p), k \in \mathbb{I}_1^N\right\}, \label{eq:RCI_definition} 
\end{align}
i.e., $\mathcal{S}(p|y_0,Y)$ is the convex hull of $N$ vertices $x^k(p)$. For each vertex, we assign a parameter-dependent vertex control input
\begin{align}
\label{eq:control_law}
    u^k(p):=u^{k,0}+U^kp, && k \in \mathbb{I}_1^N,
\end{align}
where $U^k:=[u^{k,1} \cdots u^{k,s}] \in \R^{m \times s}$. In the sequel, we derive conditions on $(y^0,Y,u^{k,0},U^k)$ to enforce $\mathcal{S}(p|y^0,Y)$ to be a PD-RCI set with vertex control inputs defined in~\eqref{eq:control_law}. 

\begin{remark}
    Assumption~\ref{ass:positive_param} is without loss of generality. For any bounded parameter set $\hat{\mathcal{P}}$ violating Assumption~\ref{ass:positive_param}, there exists a vector $\mathring{p}$ such that $p:=\mathring{p}+\hat{p} \geq \0$ for all $\hat{p} \in \hat{\mathcal{P}}$. Then, Assumption~\ref{ass:positive_param} is satisfied by the parameter set $\mathcal{P}=\{1\} \times \{ \{\mathring{p}\} \oplus \mathcal{P}\} \subset \R^{s+1}$. The method we develop can then be applied to an LPV system with 
$$A(p):=1 \left(-\underset{j=1}{\overset{s}{\sum}} \mathring{p}_j A^j\right)+\underset{j=1}{\overset{s}{\sum}} p_j A^j$$ and matrix $B(p)$ defined similarly. $\hfill\square$   
\end{remark}

\subsection{Enforcing Inclusions~\eqref{eq:LPV_conditions:1} and~\eqref{eq:RCI_vert_condition}}\label{sec:inclusions}
We will now enforce that the set $\mathcal{S}(p|y^0,Y)$ is PD-RCI under the control law defined in~\eqref{eq:cc_control_1} with vertex control inputs parameterized as in~\eqref{eq:control_law}. To this end, we first ensure that the vertex representation~\eqref{eq:RCI_definition} of the set $\mathcal{S}(p|y^0,Y)$ holds for each $p \in \mathcal{P}$ by enforcing
\begin{align}
\label{eq:config_constraint_basic}
    y^j \in \mathbb{S}, && \forall  j \in \mathbb{I}_0^s.
\end{align}
Then, the PD-RCI condition~\eqref{eq:RCI_vert_condition} in Proposition~\ref{prop:PD_RCI_basic_cond} holds if and only if for each $k \in \mathbb{I}_1^N$ and $p \in \mathcal{P}$, the inclusion
\begin{align}
    &\{A(p)x^k(p)+B(p)u^{k}(p)\}\oplus \mathcal{W} \subseteq \underset{p^+ \in \mathbb{P}(p)}{\bigcap} \mathcal{S}(p^+|y^0,Y)  \label{eq:vertex_condition_param:1}
\end{align}
is verified, along with $\mathcal{S}(p|y^0,Y) \subseteq \mathcal{X}$ and $ u^k(p) \in \mathcal{U}$.
\subsubsection{System constraints}
\noindent
To enforce state constraints, recall that $\scalemath{0.95}{\mathcal{X}=\{x:H^x x \leq h^x\}}$, and that~\eqref{eq:RCI_definition} holds under~\eqref{eq:config_constraint_basic}.
%from~\eqref{eq:RCI_definition} that $\mathcal{S}(p|y^0,Y)=\mathrm{CH}\{V^ky^0+V^kYp, k \in \mathbb{I}_1^N\}$ if the condition in~\eqref{eq:config_constraint_basic} is enforced, 
Then, the inclusion $\mathcal{S}(p|y^0,Y) \subseteq \mathcal{X}$ for all $p \in \mathcal{P}$ holds if and only if
\begin{align}
\label{eq:state_incl_1}
H^xV^ky^0 + H^xV^k Y p \leq h^x, \ \forall  \ p \in \mathcal{P}, \ \forall \  k \in \mathbb{I}_1^N.
\end{align}
Recalling that $\scalemath{0.95}{\mathcal{P}=\{p:H^p p \leq h^p\}}$, Proposition~\ref{prop:strong_duality}  states that inequality~\eqref{eq:state_incl_1} holds if and only if the following conditions are feasible:
\begin{align}
    \label{eq:state_incl_final}
    \hspace{-3pt}
    \forall \ k \in \mathbb{I}_1^N  
   \begin{cases} H^x V^k y^0+\Lambda^k h^p \leq h^x,  \\
   \Lambda^k H^p = H^x V^k Y, \\
   \Lambda^k \geq \0^{m_x \times m_p},
   \end{cases}
\end{align}
where the inequalities are enforced elementwise.
We now enforce the input constraints $\mathcal{U}=\{u:H^u u \leq h^u\}$. We note from~\eqref{eq:control_law} that $u^k(p) \in \mathcal{U}$ for all $p \in \mathcal{P}$ if and only if
\begin{align*}
    H^u u^{k,0} +H^u U^k p \leq h^u, \ \forall \ p \in \mathcal{P}, \ \forall \ k \in \mathbb{I}_1^N,
\end{align*}
that can be equivalently written by Proposition~\ref{prop:strong_duality}, as
\begin{align}
    \label{eq:input_incl_final}
    \hspace{-4pt}
    \forall \ k \in \mathbb{I}_1^N  
   \begin{cases} H^u u^{k,0}+M^k h^p \leq h^u,  \\
   M^k H^p = H^u U^k, \\
   M^k \geq \0^{m_u \times m_p}.
   \end{cases}
\end{align}
\subsubsection{PD-RCI constraints}

\begin{figure*}[ht]
\begin{equation}
\tag{28.5}
\label{eq:matrix_definitions_gen}
\mathbf{F}_{ik}(y^0,Y,u^{k,0},U^k,\munderbar{y},\munderbar{Y}):=\begin{bmatrix} -2\left(M^{ik} \begin{bmatrix} Y \\ U^k \end{bmatrix}\right) & -\left(M^{ik} \begin{bmatrix} y^0 \\ u^{k,0} \end{bmatrix}-\munderbar{Y}_i^{\top}\right) \vspace{5pt}\\
    * & 2(\munderbar{y}_i-d_i) \end{bmatrix}, \qquad \mathbf{G}(\Gamma):=\begin{bmatrix} H^{p^{\top}} \Gamma H^p & -H^{p^{\top}} \Gamma h^p \\ * & h^{p^{\top}} \Gamma h^p \end{bmatrix}
\end{equation}
\vspace{-10pt} % Adjust vertical space as needed
\noindent\rule{\textwidth}{0.4pt} % Horizontal line
\end{figure*}
To encode inclusion~\eqref{eq:vertex_condition_param:1}, we present an approach to characterize the intersection set on the right-hand-side. Defining the parameterized set $\mathcal{S}(p|\munderbar{y},\munderbar{Y})=\{x:Cx \leq \munderbar{y}+\munderbar{Y}p\}$, we enforce the inclusion
\begin{align}
    \label{eq:inside_set_inclusion}
    \mathcal{S}(p|\munderbar{y},\munderbar{Y}) \subseteq \underset{p^+ \in \mathbb{P}(p)}{\bigcap} \mathcal{S}(p^+|y^0,Y).
\end{align}
If for each $k \in \mathbb{I}_1^N$ and $p \in \mathcal{P}$, the inclusion
\begin{align}
\label{eq:modified_RCI_inside}
\{A(p)x^k(p)+B(p)u^{k}(p)\}\oplus \mathcal{W} \subseteq \mathcal{S}(p|\munderbar{y},\munderbar{Y})
\end{align}
is verified, then the desired inclusion~\eqref{eq:vertex_condition_param:1} is enforced. In order to encode inclusion~\eqref{eq:inside_set_inclusion}, we use the following result.
\begin{proposition}
\label{prop:proposition_intersection}
    For some $y^0 \in \R^{m_s}$ and $Y \in \R^{m_s \times s}$ such that $\mathcal{S}(p|y^0,Y)$ is nonempty for all $p \in \mathcal{P}$, define the vector $$\tilde{y}:=y^0 + \underset{p \in \mathcal{P}}{\min} \ Yp,$$
    where $\min$ is taken row-wise. Defining $\mathcal{Y}(\tilde{y}):=\{x:Cx\leq \tilde{y}\}$ and $\tilde{\mathcal{Q}}:=\bigcap_{p \in \mathcal{P}} \mathcal{S}(p|y^0,Y)$, it then holds that $\mathcal{Y}(\tilde{y})=\tilde{\mathcal{Q}}$. $\hfill\square$
\end{proposition}
\begin{proof}
    $1)$ For any $x \in \mathcal{Y}(\tilde{y})$, the inequality $Cx \leq \tilde{y}$ holds. As per the definition of $\tilde{y}$, this implies that $Cx \leq y^0+Yp$ holds for all $p \in \mathcal{P}$, such that $x \in \tilde{\mathcal{Q}}$. Hence, the inclusion $\mathcal{Y}(\tilde{y}) \subseteq \tilde{\mathcal{Q}}$ follows; \\
    $2)$ For any $x \in \tilde{\mathcal{Q}}$, the inequality $Cx \leq y^0+Yp$ holds for all $p \in \mathcal{P}$, or equivalently $Cx \leq \tilde{y}$ as per the definition of $\tilde{y}$. Hence, $x \in \mathcal{Y}(\tilde{y})$, such that the inclusion
    $\tilde{\mathcal{Q}} \subseteq \mathcal{Y}(\tilde{y})$ follows.
\end{proof}

Proposition~\ref{prop:proposition_intersection} implies that for a given $p \in \mathcal{P}$, the inclusion in~\eqref{eq:inside_set_inclusion} holds if and only if the inequality
\begin{align}
    \label{eq:prop_impliciation_incl}
    \munderbar{y}+\munderbar{Y}p \leq y^0+Yp^+, && \forall \ p^+ = p+\tilde{p} \in \mathbb{P}(p).
\end{align}
In order to encode~\eqref{eq:prop_impliciation_incl} for all $p \in \mathcal{P}$, we define the set
\begin{align*}
    \mathcal{P}^+ := \left\{ \begin{bmatrix} p \\ \tilde{p} \end{bmatrix} : \underbrace{\begin{bmatrix} H^p & \0 \\ \0 & H^{\delta} \\ H^p & H^p \end{bmatrix}}_{H^{p\delta}} \begin{bmatrix} p \\ \tilde{p} \end{bmatrix} \leq \underbrace{\begin{bmatrix} h^p \\ h^{\delta} \\ h^p \end{bmatrix}}_{h^{p\delta}} \right\}.
\end{align*}
Then, inequality~\eqref{eq:prop_impliciation_incl} holds for all $p \in \mathcal{P}$ if and only if
\begin{align}
\label{eq:inequality_intersection}
    \munderbar{y}+[\munderbar{Y}-Y \ \ -Y] \begin{bmatrix} p \\ \tilde{p} \end{bmatrix} \leq y^0, && \forall  \begin{bmatrix} p \\ \tilde{p} \end{bmatrix} \in \mathcal{P}^+,
\end{align}
From Proposition~\ref{prop:strong_duality}, inequality~\eqref{eq:inequality_intersection} holds if and only if the conditions
\begin{align}
\label{eq:inter_set_necc_suff}
\begin{matrix*}[l]
    \munderbar{y}+ {Q} h^{p\delta} \leq y^0, \vspace{3pt}\\
    {Q} H^{p\delta} = [\munderbar{Y}-Y \ \ -Y], \vspace{3pt}\\
    {Q}  \geq  \0^{m_s \times (2m_p+m_{\delta})}
\end{matrix*}
\end{align}
are verified.
Finally, we encode the RCI inclusion in~\eqref{eq:modified_RCI_inside}. To this end, we tighten the set $\mathcal{S}(p|\munderbar{y},\munderbar{Y})$ by the disturbance set $\mathcal{W}$ as
\begin{align*}
    \mathcal{S}(p|\munderbar{y},\munderbar{Y}) \ominus \mathcal{W}=\{x : Cx \leq {\munderbar{y}}+\munderbar{Y}p-d\}, &&  d:=\max_{w \in \mathcal{W}} C w,
\end{align*}
such that inclusion~\eqref{eq:modified_RCI_inside} holds for a given $k \in \mathbb{I}_1^N$ and $p \in \mathcal{P}$ if and only if for every row index $i \in \mathbb{I}_1^{m_s}$, the inequality
%if and only if for every $k \in \mathbb{I}_1^N$, $p \in \mathcal{P}$, the following inequality is verified for all $i \in \mathbb{I}_1^{m_s}$:
\begin{align}
\label{eq:inequality_1}
 C_i(A(p)x^k(p)+B(p)u^k(p)) \leq {\munderbar{y}}_i+\munderbar{Y}_ip-d_i
\end{align}
is satisfied.
We recall that for a given parameter $p \in \mathcal{P}$, each vertex and the corresponding control input of the set $\mathcal{S}(p)$ are given by $x^k(p)=V^k(y^0+Yp)$ and $u^k(p)=u^{k,0}+U^kp$, respectively. Then, we define matrices $\bar{C}^i \in \R^{s \times ns}, \bar{A}\in \R^{ns \times n}$ and $\bar{B} \in \R^{ns \times m}$ as
\begin{align*}
    &\bar{C}^i:=\I^s \otimes C_i, \ \
    \bar{A}:=[A^{1^{\top}} \cdot \cdot A^{s^{\top}}]^{\top}, \ \
    &\bar{B}:=[B^{1^{\top}} \cdot \cdot  B^{s^{\top}}]^{\top}, 
\end{align*}
based on which we define $M^{ik}:=\bar{C}^i[\bar{A}V^k \ \ \bar{B}] \in \R^{s \times (m_s+m)}$. Rearranging~\eqref{eq:inequality_1}, we rewrite the inequality as
\begin{align}
\label{eq:inequality_2}
\begin{bmatrix} p \\ 1 \end{bmatrix} ^{\top} \mathbf{F}_{ik}(y^0,Y,u^{k,0},U^k,\munderbar{y},\munderbar{Y}) \begin{bmatrix} p \\ 1 \end{bmatrix} \geq 0,
\end{align}
where the function $\mathbf{F}_{ik}$ is defined in~\eqref{eq:matrix_definitions_gen}.
For the PD-RCI condition in~\eqref{eq:modified_RCI_inside} to be satisfied, the parameters $\{y^0,Y,u^{k,0},U^k,k\in \mathbb{I}_1^N,\munderbar{y},\munderbar{Y}\}$ should be such that the the inequality in~\eqref{eq:inequality_2} is verified by all $p \in \mathcal{P}$ at all vertex indices $k \in \mathbb{I}_1^N$ and row indices $i \in \mathbb{I}_1^{m_s}$. Since~\eqref{eq:inequality_2} is a (non)convex quadratic inequality in $p$, we propose to use the S-procedure~\cite{lmi_book} to derive sufficient conditions to enforce it over all $p \in \mathcal{P}$. 
To this end, we recall from~\cite{Fazlyab2022} that the polytopic parameter set $\mathcal{P}=\{p:H^p p \leq h^p\}$ with $h^p \in \R^{m_p}$ satisfies
\begin{align*}
\mathcal{P} \subseteq \bigcap_{\Gamma \in \mathbb{G}} \left\{ p : \begin{bmatrix} p \\ 1 \end{bmatrix} ^{\top} \mathbf{G}(\Gamma) \begin{bmatrix} p \\ 1 \end{bmatrix} \geq 0 \right\},
\end{align*}
where the function $\mathbf{G}$ is defined in~\eqref{eq:matrix_definitions_gen}, and $\mathbb{G}$ is the set of symmetric matrices $\Gamma$ defined as
\begin{align*}
    \mathbb{G}:=\left\{\Gamma \in \R_+^{m_p \times m_p} : \Gamma=\Gamma^{\top}, \ \mathrm{diag}(\Gamma)=\0^{m_p} \right\}.
\end{align*}
Then, the inequality in~\eqref{eq:inequality_1} is verified by all $p \in \mathcal{P}$ if
\begin{align}
\label{eq:LMI_main_condition}
   \exists \ \Gamma^{k,i} \in \mathbb{G} \ : \ \mathbf{F}_{ik}(y^0,Y,u^{k,0},U^k,\munderbar{y},\munderbar{Y})-\mathbf{G}(\Gamma^{k,i}) \succeq 0,
\end{align}
which is a linear matrix inequality (LMI). Thus, by enforcing the LMI in~\eqref{eq:LMI_main_condition} for all $k \in \mathbb{I}_1^N$ and $i \in \mathbb{I}_1^{m_s}$, we obtain a convex characterization of the PD-RCI sets $\mathcal{S}(p)$. 
\begin{remark}
\label{remark:QC}
The unified parameterization of hyperplane and vertex representations via a single vector enables the synthesis of polytopic PD-RCI sets under more complex constraints. Specifically, polynomial constraints can be addressed through convex feasibility conditions derived via the S-procedure or Sum-of-Squares programming \cite{Fazlyab2022, cotoruello2021}. Developing these methods is a future research topic. $\hfill\square$
\end{remark}
\begin{remark}
Conservativeness in our approach primarily stems from fixing the normal vectors of the sets $\mathcal{S}(p|y^0,Y)$, a necessary trade-off in fixed parameterizations \cite{Gupta2019}. Further conservativeness may result from configuration constraints, with implications on polytope vertex configurations meriting future research \cite{Loechner1997}. Additionally, employing the S-procedure introduces a potential duality gap as discussed in \cite{Luo2010}, which is a known issue in deriving conditions like \eqref{eq:LMI_main_condition}.
$\hfill\square$
    \end{remark}
\noindent
\noindent
%that holds according to the $S$-procedure~\cite{lmi_book} if
%which can inturn be written using the $S$-procedure as
%Hence, $\mathcal{S}(p|y^0,Y)$ is a PD-RCI set for System~\eqref{eq:LPV_system} with PD-vertex inputs $\{u^k(p)=u^{k,0}+U^kp,k\in \mathbb{I}_1^N\}$ if the parameters $(y^0,Y,u^{k,0},U^{k}, k \in \mathbb{I}_1^N,\munderbar{y},\munderbar{Y})$  satisfy~\eqref{eq:config_constraint_basic},~\eqref{eq:state_incl_final},~\eqref{eq:input_incl_final},~\eqref{eq:inter_set_necc_suff}, and~\eqref{eq:LMI_main_condition}. 

\subsection{Maximizing the size of the PD-RCI set}\label{sec:volume}
We define the size of a set $\mathcal{Z} \subseteq \mathcal{X}$ as
\begin{align}
    \label{eq:distance_metric}
    \mathrm{d}_{\mathcal{X}}(\mathcal{Z}):=\min_{\epsilon}\{\norm{\epsilon}_1 \ \ \mathrm{s.t.} \ \ \mathcal{X} \subseteq \mathcal{Z} \oplus \mathcal{D}(\epsilon)\},
\end{align}
where $\mathcal{D}(\epsilon):=\{x : Dx \leq \epsilon\}$ is a polytope with user-specified normal vectors $\{D^{\top}_i,i \in \mathbb{I}_1^{m_d}\}$. This is a modification of the Hausdorff distance between $\mathcal{Z}$ and $\mathcal{X}$:  if $\mathcal{D}(\epsilon)=\epsilon \mathcal{B}_{l}^n$, then $\mathrm{d}_{\mathcal{X}}(\mathcal{Z})$ is the standard $l$-norm Hausdorff distance. By allowing $D$ to be user-specified, we permit maximization in directions of interest. Clearly, $\mathrm{d}_{\mathcal{X}}(\mathcal{Z}) \geq 0$, and $\mathcal{Z}_1 \subseteq \mathcal{Z}_2 \subseteq \mathcal{X}$ implies
$\mathrm{d}_{\mathcal{X}}(\mathcal{Z}_2) \leq \mathrm{d}_{\mathcal{X}}(\mathcal{Z}_1)$.
%using a modification of the Hausdorff distance between $\mathcal{Z}$ and $\mathcal{X}$. In particular, using some polytope $\mathcal{D}(\epsilon):=\{x : Dx \leq \epsilon\}$
%\begin{align}
%\label{eq:distance_matrix}
%    \mathcal{D}(\epsilon):=\{x : Dx \leq \epsilon\} \ \text{where $D \in \R^{m_d \times n}$},
%\end{align}
%with the vectors $\{D^{\top}_i,i \in \mathbb{I}_1^{m_d}\}$ chosen a priori, we use 
%as a measure of size of $\mathcal{Z}$ with respect to $\mathcal{X}$. 
%Clearly, $\mathrm{d}_{\mathcal{X}}(\mathcal{Z}) \geq 0$, and for any $\mathcal{Z}_1 \subseteq \mathcal{Z}_2 \subseteq \mathcal{X}$, the inequality $\mathrm{d}_{\mathcal{X}}(\mathcal{Z}_2) \leq \mathrm{d}_{\mathcal{X}}(\mathcal{Z}_1)$ holds. Note that if $\mathcal{D}(\epsilon)=\epsilon \mathcal{B}_{l}^n$, then $\mathrm{d}_{\mathcal{X}}(\mathcal{Z})$ is the standard $l$-norm Hausdorff distance. By allowing $D$ to be user-specified, we permit maximization in directions of interest.
Ideally, we want to compute the PD-RCI set that minimizes $$\sum_{p \in \mathcal{P}} \mathrm{d}_{\mathcal{X}}(\mathcal{S}(p|y^0,Y))$$
%\begin{align}
%\label{eq:exact_objective}
%    \sum_{p \in \mathcal{P}} \mathrm{d}_{\mathcal{X}}(\mathcal{S}(p|y^0,Y)),
%\end{align}
to maximize the invariant region for each $p \in \mathcal{P}$. Unfortunately, this objective is infinite dimensional. As an alternative, we sample parameters $\{p^j \in \mathcal{P},j \in \mathbb{I}_1^{\theta}\}$, and minimize its lower-bound 
%$ \sum_{j=1}^{\theta} \mathrm{d}_{\mathcal{X}}(\mathcal{S}(p^j|y^0,Y)).$
\begin{align}
\label{eq:approximate_objective}
    \sum_{j=1}^{\theta} \mathrm{d}_{\mathcal{X}}(\mathcal{S}(p^j|y^0,Y)).
\end{align}
For example, points $p^j$ could be the vertices of $\mathcal{P}$. To implement~\eqref{eq:approximate_objective}, we must encode the inclusions
\begin{align}
\label{eq:cost_inclusions}
\mathcal{X}\subseteq \mathcal{S}(p^j|y^0,Y) \oplus \mathcal{D}(\epsilon^j), && \forall  j \in \mathbb{I}_1^{\theta}.
\end{align}
To this end, we assume to be given the vertices $\{\mathrm{x}^t,t \in \mathbb{I}_1^{v_x}\}$ of  $\mathcal{X}$. Then, the inclusions in~\eqref{eq:cost_inclusions} are equivalent~\cite{Schneider2013} to
%(\MM{cite reference if any?})
\begin{align}
\label{eq:approx_constraints}
\forall \ j \in \mathbb{I}_1^{\theta}, \ t \in \mathbb{I}_1^{v_x}  \begin{cases}  \mathrm{x}^t= \mathrm{s}^{t,j}+ \mathrm{b}^{t,j}, \\
\mathrm{s}^{t,j} \in \mathcal{S}(p^j|y^0,Y), \ \ \mathrm{b}^{t,j} \in \mathcal{D}(\epsilon^j). 
\end{cases}
\end{align}
%
%\begin{align}
%\label{eq:approx_constraints}
%\forall  j \in \mathbb{I}_1^{\theta}, t \in \mathbb{I}_1^{v_x}, \exists begin{cases} \ \mathrm{s}^{t,j} \in \mathcal{S}(p^j|y^0,Y), \ \mathrm{b}^{t,j} \in \mathcal{D}(\epsilon^j): \vspace{3pt} \\ \qquad \qquad \mathrm{x}^t= \mathrm{s}^{t,j}+ \mathrm{b}^{t,j}. \end{cases}
%\end{align}
The inclusion can also be encoded directly with a halfspace representation of $\mathcal{X}$ using~\cite{Sadraddini2019}.
%If the vertices of $\mathcal{X}$ are unavailable, then the inclusion can be encoded with the halfspace notation %using~\cite{Sadraddini2019}.
Thus, a large PD-RCI set $\mathcal{S}(p|y^0,Y)$ for LPV system~\eqref{eq:LPV_system} with the PD-vertex control law in~\eqref{eq:control_law} can be computed by solving the SDP problem
\begin{align}
\label{eq:final_SDP}
\min_{\mathbf{x}} \ \ \sum_{j=1}^{\theta} \norm{\epsilon^j}_1 \ \
\text{s.t.} \ \ \eqref{eq:config_constraint_basic}, \eqref{eq:state_incl_final}, \eqref{eq:input_incl_final}, \eqref{eq:inter_set_necc_suff}, \eqref{eq:LMI_main_condition}, \eqref{eq:approx_constraints}
\end{align}
with the optimization variables
\begin{align*}
    \mathbf{x}:=\begin{Bmatrix} y^0,Y,\{u^{k,0},U^k,k \in \mathbb{I}_1^N\},\munderbar{y},\munderbar{Y},  \{\Lambda^k,M^k,k \in \mathbb{I}_1^N\}, Q,\\
    \{\Gamma^{k,i}, k \in \mathbb{I}_1^N, i \in \mathbb{I}_1^{m_s}\},
    \{\epsilon^j,\mathrm{s}^{t,j},\mathrm{b}^{t,j}, j \in \mathbb{I}_1^{\theta}, t \in \mathbb{I}_1^{v_x}\}
    \end{Bmatrix}.
\end{align*}

\subsection{Selecting the matrix $C$ parameterizing the PD-RCI set}
\label{sec:selecting_C}
We briefly discuss some methods for choosing a matrix $C$ to parameterize the PD-RCI set $\mathcal{S}(p|y^0,Y)$ that guarantee the feasibility of Problem~\eqref{eq:final_SDP}. Notably, if the set $\mathcal{Y}(y_{\mathrm{PI}})=\{x:Cx \leq y_{\mathrm{PI}}\}$ is a parameter-independent (PI) RCI set for~\eqref{eq:LPV_system} for some $y_{\mathrm{PI}}$, then Problem~\eqref{eq:final_SDP} is feasible with $y^0=y_{\mathrm{PI}}$ and $Y=\0$.

This implies that the normal vector matrix of any polytopic RCI set can serve as matrix C. Established techniques like \cite{Bertsekas72,Pluymers2005_LPV} can compute such a matrix, but may lead to RCI sets with a large number of hyperplanes and vertices, making Problem~\eqref{eq:final_SDP} computationally expensive. Alternatively, methods allowing an a priori specification of representational complexity, e.g.,~\cite{Gupta2019,Liu2019}, can be used to obtain matrix $C$. However, such approaches can result in conservative PD-RCI sets. In the numerical examples in Section~\ref{sec:examples}, we compare these approaches to demonstrate the importance of selecting a suitable matrix $C$ that balances between representational complexity and conservativeness. We now present a method to calculate a candidate RCI set which we have empirically observed to balance complexity and conservativeness. This approach is a simplified variant of the one in~\cite{Gupta2020_vertex}, focusing exclusively on the linear dependency of system matrices on the parameter as in~\eqref{eq:P_dependence}, as opposed to rational dependency.

We propose to compute a PI-RCI set parameterized as
\begin{align}
    \label{eq:gupta_vertex_set}
    \mathcal{S}_{\mathrm{PI}}(\bm{W}):=\{x : \hat{C}\bm{W}^{-1} x \leq \1^{m_s}\}, && \hat{C} \in \R^{m_s \times n},
\end{align}
where the matrix $\hat{C}$ is selected a priori, and the set is parameterized by the invertible matrix $\bm{W} \in \R^{n \times n}$. Following the approach of~\cite{Gupta2020_vertex}, we transform the state of~\eqref{eq:LPV_system} as $\bm{z}=\bm{W}^{-1}x$, such that the dynamics in the transformed state space are 
\begin{align}
    \label{eq:transformed_space}
    \mathbf{z}^+=\bm{W}^{-1}(A(p)\bm{W} \mathbf{z}+ B(p) u+w).
\end{align}
Then, as shown in~\cite[Lemma 2]{Gupta2020_vertex}, if the set
\begin{align}
\mathbf{Z}:=\{\mathbf{z} : \hat{C} \mathbf{z} \leq \1\}=\mathrm{CH}\{\hat{\mathbf{z}}^j, j \in \mathbb{I}_1^N\}
\end{align}
is RCI for~\eqref{eq:transformed_space}, then the set $\mathcal{S}_{\mathrm{PI}}(\bm{W})$ is RCI for~\eqref{eq:LPV_system}. Since the vertices of $\mathbf{Z}$ are known a priori (as $\hat{C}$ is selected a priori), the RCI condition can be enforced by associating to each vertex a feasible control input $\{u^j \in \mathcal{U}, j \in \mathbb{I}_1^N\}$, and enforcing the inequality
\begin{align}
\label{eq:ankit_RCI_2}
    \hat{C}\bm{W}^{-1}(A(p)\bm{W} \hat{\mathbf{z}}^j + B(p) u^j + w) \leq \1
\end{align}
at each vertex index $j \in \mathbb{I}_1^N$ for all parameters $p \in \mathcal{P}$ and disturbances $w \in \mathcal{W}$.
Since the inequality in~\eqref{eq:ankit_RCI_2} is linear in $(A(p),B(p))$ for given $\bm{W}$, and $(A(p),B(p))$ depend linearly on $p$ as in~\eqref{eq:P_dependence}, it can be enforced for all parameters $p \in \mathcal{P}$ by enforcing it for all $p=p^i$, where $\{p^i, i \in \mathbb{I}_1^q\}$ are the vertices of $\mathcal{P}$. Similarly, the inequality can be enforced for all $w \in \mathcal{W}$ by enforcing it for all $w = w^l$, where $\{w^l,l\in \mathbb{I}_1^{q_w}\}$ are the vertices of $\mathcal{W}$. Note that if the vertices of $\mathcal{P}$ and $\mathcal{W}$ are not available, then Proposition~\ref{prop:strong_duality} can be used to enforce~\eqref{eq:ankit_RCI_2} directly using their hyperplane representations.  
 %the condition $\forall  p \in \mathcal{P}$ can be eliminated by enforcing the inequality for all $\tilde{A}^i:=A([\hat{p}^i \ 1]^{\top})$ and $\tilde{B}^i:=B([\hat{p}^i \ 1]^{\top})$ with $i \in \mathbb{I}_1^4$, i.e., LPV system matrices corresponding to the vertices of $\mathcal{P}$. In order to eliminate condition $\forall w \in \mathcal{W}$, we use Proposition~\ref{prop:strong_duality}. In particular, we introduce the multiplier matrix $\bm{\Lambda}^{j,i} \in \R_+^{m_s \times m_w}$, such that the RCI condition can be written as
 %\begin{align}
 \begin{figure}[t]
\centering
\vspace*{0.1cm}
\hspace*{-0.55cm}
{
\resizebox{0.55\textwidth}{!}
{
\begin{tikzpicture}
\begin{scope}[xshift=0.cm]
\node[draw=none,fill=none](tanks_fig) {\includegraphics[trim=35 0 0 0,clip,scale=3]{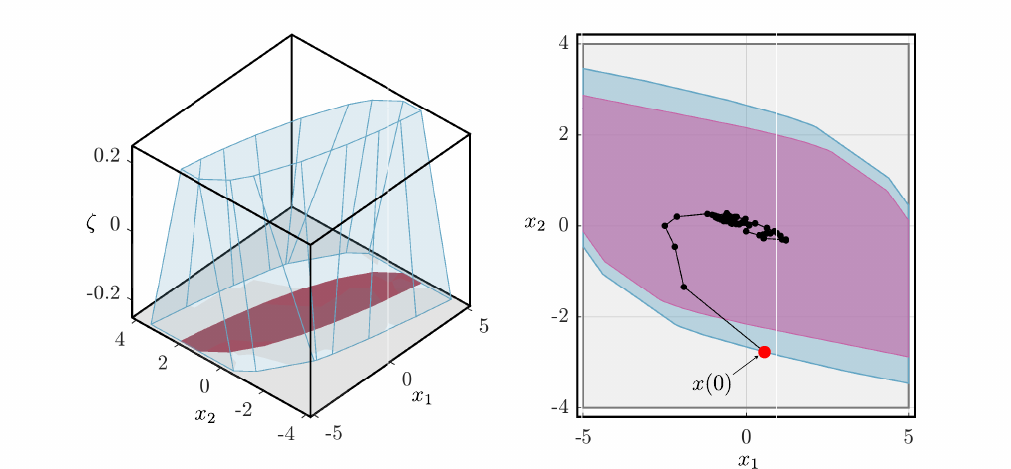}};
\end{scope}
\end{tikzpicture}}
}
\captionsetup{width=1\linewidth}
\caption{Results for Example~\ref{sec:Example1}. (\textit{Left: ($x$-$\zeta$ space}) The blue set is $\{(x,\zeta) : \zeta \in [-0.25,0.25], x \in \mathcal{S}([0.5+2\zeta,0.5-2\zeta])\}$, and the red set is $X_{\infty} \times \{-0.25\}$, where $X_{\infty}$ is the MRCI set. We obtain larger RCI sets by explicitly accounting for parameter variation. (\textit{Right: $x$ space}) The gray set is $\mathcal{X}$, and the blue set is $\mathcal{S}(p)$ with $p=[0,1]$ corresponding to $\zeta=-0.25$. The pink set is $\cap_{p^+ \in \mathbb{P}(p)} \mathcal{S}(p^+)$. Initializing $p(0) = [0,1]$ and $x(0) \in \mathcal{S}(p(0))$, we have $x(1) \in \cap_{p^+ \in \mathbb{P}(p)} \mathcal{S}(p^+)$ with $u(0)$ computed as~\eqref{eq:cc_control_1}. The black dotted line is the simulation trajectory, obtained by randomly sampling $p(t)$ while enforcing satisfaction of~\eqref{eq:bounded_param_variation}, and disturbance $w(t)$ sampled randomly from the vertices of $\mathcal{W}$.}
\label{fig:Example_1_results}
%\vspace*{-0.7cm}
\end{figure}

 %\label{eq:ankit_RCI_3}
 %\forall  i \in \mathbb{I}_1^4, \ j \in \mathbb{I}_1^N, \begin{cases}
 %\exists  \bm{\Lambda}^{j,i} \in \R_+^{m_s \times m_w}, \ \exists  u^j \in \mathcal{U} : \\
  %   \hat{C}\bm{W}^{-1}(\tilde{A}^i \bm{W} \hat{\mathbf{z}}^j + \tilde{B}^i u^j) + \bm{\Lambda}^{j,i} h^w \leq \1 \\
 %    \bm{\Lambda}^{j,i} H^w = \hat{C}\bm{W}^{-1} B_w.
  %   \end{cases}
 %\end{align}
In order to compute the RCI parameters $\bm{W}$ and $\{u^{j}, j \in \mathbb{I}_1^N\}$, we introduce the matrix $\bm{M}$ in lieu of $\bm{W}^{-1}$   as a variable in constraint~\eqref{eq:ankit_RCI_2}, and introduce the constraint $\bm{W}\bm{M} = \I^n$. Then, we formulate and solve the nonlinear programming problem (NLP) 
\begin{subequations}
\label{eq:NLP_init_C}
\begin{align}
    &\hspace{-10pt} \min_{\bm{W},\bm{M},\{u^j, j \in \mathbb{I}_1^N\}}   \ \ \mathrm{d}_{\mathcal{X}}(\mathcal{S}_{\mathrm{PI}}(\bm{W})) \\
    & \hspace{15pt} \text{s.t.} \qquad  \hat{C}\bm{M}(A(p^i)\bm{W} \hat{\mathbf{z}}^j + B(p^i) u^j + w^l) \leq \1, \\
    & \hspace{43pt}  H^x \bm{W} \hat{\mathbf{z}}^j \leq h^x, \ H^u u^j \leq h^u,  \  \bm{W}\bm{M} = \I^n, \label{eq:NLP_init_C:state_con} \\
     & \hspace{43pt}  \forall  \ j \in \mathbb{I}_1^N, \ i \in \mathbb{I}_1^q, \ l \in \mathbb{I}_1^{q_w}, 
\end{align}
\end{subequations}
where~\eqref{eq:NLP_init_C:state_con} enforces  the state constraints $\mathcal{S}_{\mathrm{PI}}(\bm{W}) \subseteq \mathcal{X}$, and the objective defined in~\eqref{eq:distance_metric} minimizes the distance between $\mathcal{S}_{\mathrm{PI}}(\bm{W})$ and $\mathcal{X}$.
Using the solution of~\eqref{eq:NLP_init_C}, we parameterize the PD-RCI set in~\eqref{eq:PD_RCI_parameterization} as $C  \leftarrow \hat{C} \bm{W}^{-1}$.
Note that Problem~\eqref{eq:NLP_init_C} can be solved using a standard off-the-shelf NLP solver like IPOPT~\cite{Wchter2005}.

The following scheme summarizes our approach to synthesize and use PD-RCI sets and PD-vertex control laws for System~\eqref{eq:LPV_system}.
\begin{enumerate}
    \item Select matrix $C$ parameterizing an RCI set as in \eqref{eq:PD_RCI_parameterization}. Alternatively, compute appropriate matrix $C$ by solving Problem~\eqref{eq:NLP_init_C}.
    \item Construct configuration constraints matrices $\mathbf{E}$ and $V^k$ as in Appendix~\ref{sec:appendix}.
    \item Select matrix $D$ to formulate the distance function in~\eqref{eq:distance_metric}.
    \item Solve the SDP~\eqref{eq:final_SDP}, extract PD-RCI set parameters $y^0,Y$ and PD-vertex controls $\{u^{k,0},U^k,k \in \mathbb{I}_1^N\}$.
    \item If $Cx(t)\leq y^0+Yp(t)$ holds, apply the input $u(t)$ in~\eqref{eq:cc_control_1}.
\end{enumerate}

\section{Numerical examples}
\label{sec:examples}
In this section, we present three numerical examples to demonstrate our PD-RCI set computation approach. In Examples~\ref{sec:Example1} and~\ref{sec:Example2}, we compare our approach with the method proposed in~\cite{Gupta2023} for computing PD-RCI sets.
We recall that in~\cite{Gupta2023}, the PD-RCI set is parameterized as the $0$-symmetric polytope
 \begin{align}
     \label{eq:ankit_RCI}
     \mathcal{S}(p)=\left\{x: -\1 \leq \left(\sum_{j=1}^s p_j P^j\right) W^{-1} x \leq \1\right\},
 \end{align}
that is rendered positive invariant with the PD-linear feedback law $ u=\left(\sum_{j=1}^s p_jK^j\right)x.$
%\begin{align}
%     \label{eq:ankit_control}
%     u=\left(\sum_{j=1}^s p_jK^j\right)x.
%\end{align}
The parameters $\{P^j,K^j,j \in \mathbb{I}_1^s,W\}$ are computed using a sequential SDP methodology.  In Example~\ref{sec:Example3}, we employ the NLP method outlined in Section~\ref{sec:selecting_C} to compute a matrix $C$ that parameterizes a PD-RCI set for a 4-dimensional system.
The SDP problems were modeled using YALMIP~\cite{Lofberg2004}, and solved using the MOSEK SDP solver~\cite{mosek} on a laptop with Intel i7-7500U CPU and 16GB of RAM. The NLP in Example~\ref{sec:Example3} is modeled using CasADI~\cite{Andersson2019} and solved using IPOPT~\cite{Wchter2005}.

\subsection{Double integrator}
\label{sec:Example1}
We consider the parameter-varying double integrator
\begin{align*}
    x^+=\begin{bmatrix} 1+\zeta & 1+\zeta \\ 0 & 1+\zeta \end{bmatrix} x+\begin{bmatrix} 0 \\ 1+\zeta \end{bmatrix} u+w, && |\zeta|\leq 0.25, 
\end{align*}
with $\mathcal{X}=5\mathcal{B}_{\infty}^2$, $\mathcal{U}=\mathcal{B}_{\infty}^1$ and $\mathcal{W}=\{w:|w| \leq [0.25\ 0]^{\top}\}$.
This system can be brought to the form in~\eqref{eq:LPV_system} with 
$$
\scalemath{0.92}{
\left[ \begin{array}{c|c}
   A^1 &
   A^2
\end{array}\right]=\left[ \begin{array}{c|c}
   \begin{matrix} 1.25 & 1.25 \\ 0 & 1.25 \end{matrix}  & \begin{matrix} 0.75 & 0.75 \\ 0 & 0.75 \end{matrix}
   \end{array} \right], \ 
   \begin{matrix} B^1=\begin{bmatrix} 0 & 1.25 \end{bmatrix}^{\top}, \\ B^2=\begin{bmatrix} 0 & 0.75 \end{bmatrix}^{\top}, \end{matrix}
   }
$$
%\begin{align*}
%\scalemath{0.95}{
%    A^1=\begin{bmatrix} 1.25 & 1.25 \\ 0 & 1.25 \end{bmatrix}, A^2=\begin{bmatrix} 0.75 & 0.75 \\ 0 & 0.75 \end{bmatrix}, \begin{matrix} B^1=\begin{bmatrix} 0 & 1.25 \end{bmatrix}^{\top}, \\ B^2=\begin{bmatrix} 0 & 0.75 \end{bmatrix}^{\top}, \end{matrix}}
%\end{align*}
%$B^1=\begin{bmatrix} 0 & 1.25 \end{bmatrix}^{\top}$, and $B^2=\begin{bmatrix} 0 & 0.75 \end{bmatrix}^{\top}$
using $p=[(0.5+2\zeta),(0.5-2\zeta)]$ 
and the simplex parameter set $\mathcal{P}=\{p:p \in [\0,\1], p_1+p_2=1\}$.
\begin{table}[t]
\centering
\begin{tabular}{|c|c|c|c|c|c|}
\hline
$\zeta$ & -0.25 & -0.125 & 0 & 0.125 & 0.25 \\ \hline
$\mathrm{d}_{\mathcal{X}}(\mathcal{S}(p))$ & 37.1124 & 41.2856 & 45.4587 & 49.6318 & 53.8049 \\ \hline
\end{tabular}
\caption{Values of $\mathrm{d}_{\mathcal{X}}(\mathcal{S}(p))$ as a function of scheduling parameter $\zeta$.}
\label{tab:dx_sp}
\end{table}

% Second Table
\begin{table}[t]
\centering
\begin{tabular}{|c|c|c|c|c|c|c|c|}
\hline
$\kappa$ & 0.05 & 0.1 & 0.2 & 0.3 & 0.4 & 0.5 \\ \hline
$d_{\mathrm{tot}}$ & 261.37 & 265.64 & 273.78 & 274.76 & 274.94 & 274.94 \\ \hline
\end{tabular}
\caption{Total distance $d_{\mathrm{tot}}$ for different values of parameter rate variation bound $\kappa$.}
\label{tab:dtot_kappa}
\end{table}
 In Figure~\ref{fig:Example_1_results}, we plot the PD-RCI set computed for this system using $C \in \R^{16 \times 2}$ constructed using the normal vectors to the maximal RCI (MRCI) set $X_{\infty}$, for which we compute the matrices $\{V^k,k\in \mathbb{I}_1^{16},\mathbf{E}\}$ as described in \cite[Section 3.5]{villanueva2022configurationconstrained}. For simplicity, we select $\mathcal{D}(\epsilon)=\{x:Cx \leq \epsilon\}$ in the formulation of the constraint in~\eqref{eq:approx_constraints}. Finally, we use the parameter variation bound $\mathcal{R}=0.2 \mathcal{B}^2_{\infty}$. We report that $\mathrm{d}_{\mathcal{X}}(X_{\infty})=53.20$ (see Equation~\eqref{eq:distance_metric} for definition), while the size of the PD-RCI set varies with $\zeta$, recalling that $p=[(0.5+2\zeta),(0.5-2\zeta)]$, as in Table~\ref{tab:dx_sp}.
%\begin{align*}
%    \begin{array}{cccccc}
%    \hline
%    \zeta & -0.25 & -0.125 & 0 & 0.125 & 0.25 \\
%    \hline
%    \mathrm{d}_{\mathcal{X}}(\mathcal{S}(p)) & 37.1124  & 41.2856  & 45.4587  & 49.6318  & 53.8049 \\
%    \hline
%    \end{array}
%\end{align*}
Since smaller values of $\mathrm{d}_{\mathcal{X}}$ correspond to greater coverage of $\mathcal{X}$ by the PD-RCI set, we observe reduced conservativeness by explicitly accounting for the scheduling parameter. We report that total construction and solution time is $2.5$s.

\textbf{\textit{Comparison with~\cite{Gupta2023}:}} We now compare our results with those obtained using the approach presented in~\cite{Gupta2023}. To this end, we parameterize our matrix $C$ using the solution of the procedure in~\cite{Gupta2023} with $p=[0.5,0.5]$, resulting in $m_s=N=8$. We also select $D=C$. Since an advantage of our approach compared to~\cite{Gupta2023} is the ability to explicitly account for parameter variation bound $\mathcal{R}$, we perform the comparison utilizing $\mathcal{R}=\kappa \mathcal{B}_{\infty}^2$ for different values of $\kappa$. We compare the result using the metric
\begin{align*}
    \mathrm{d}_{{\mathrm{tot}}}:=\sum_{p \in \tilde{\mathcal{P}}} d_{\mathcal{X}}(\mathcal{S}(p)),
\end{align*}
where $\tilde{\mathcal{P}} \subset \mathcal{P}$ is a discrete valued set of parameters sampled from $\mathcal{P}$. In our experiments, we build $\tilde{\mathcal{P}}$ using $200$ samples of $\zeta$ sampled uniformly in $[-0.25,0.25]$. The resulting value of $\mathrm{d}_{\mathrm{tot}}$ for the set $\mathcal{S}(p)$ obtained in~\cite{Gupta2023} is $277.01$, while as we vary $\kappa$, we obtain the values in Table~\ref{tab:dtot_kappa},
%\begin{align*}
%    \begin{array}{ccccccccccc}
%    \hline
%    & \kappa & 0.05& 0.1 & 0.2 & 0.3 & 0.4 & 0.5 \\
%    \hline
%    &d_{\mathrm{tot}} & 261.37 & 265.64 &273.78& 274.76 & 274.94 & 274.94  \\
%    \hline
%    \end{array}
%\end{align*}
and $d_{\mathrm{tot}} = 274.94$ for all $\kappa \in [0.5,1]$. As expected, these results indicate reduced conservativeness in the PD-RCI sets with respect to the approach of~\cite{Gupta2023} when explicitly accounting for the parameter variation bounds. 
\subsection{Quasi-LPV system}
\label{sec:Example2}
 As in~\cite{Gupta2023}, our approach can be used to synthesize RCI sets for nonlinear systems that can be represented as quasi-LPV systems in which the scheduling parameter depends on the current state. Because of this dependency, the RCI set representation must be independent of the scheduling parameter. However, invariance in the RCI set can be enforced using a parameter-dependent control law. As compared to~\cite{Gupta2023}, we enforce invariance using a parameter-dependent vertex control law instead of a parameter-dependent linear feedback law. Since any linear feedback law can be interpolated by a vertex control law~\cite{Gutman1986}, it follows that the set we compute is less conservative. For illustration, we consider the Van der Pol oscillator system
\begin{align*}
    \dot{x}_1=x_2, \ \  
    \dot{x}_2=-x_1+\mu(1-x_1^2)x_2+u.
\end{align*}
We discretize this system using the forward Euler scheme with timestep $\delta$, in order to obtain the discrete time dynamics
\begin{align*}
x(t+1)=\left(\begin{bmatrix} 1 & \delta \\ -\delta & 1 \end{bmatrix}+\begin{bmatrix} 0 & 0 \\ 0 & \mathbf{q}(x_1(t)) \end{bmatrix}\right) x(t)+\begin{bmatrix} 0 \\ \delta \end{bmatrix} u(t),
\end{align*}
where $\mathbf{q}(x_1)=\mu\delta(1-x_1^2)$. As described in~\cite{Gupta2023}, we obtain an LPV representation~\eqref{eq:LPV_system} for this system with
$$
\left[ \begin{array}{c|c}
   A^1 &
   A^2 
\end{array}\right]=\left[ \begin{array}{c|c}
   \begin{matrix} 1 & \delta \\ -\delta & 1 \end{matrix} &     \begin{matrix} 1 & \delta \\ -\delta & 2 \end{matrix} 
\end{array}\right], \ \ B^1,B^2=\begin{bmatrix} 0 \\ \delta \end{bmatrix},
$$
by selecting $p_1=1-\mathbf{q}(x_1)$ and $p_2=\mathbf{q}(x_1)$. The system constraints are $\mathcal{X}=\{x:\norm{x}_{\infty} \leq 1\}$ and $\mathcal{U}=\{u:|u|\leq 1\}$. Then, $\mathcal{P}=\{p:p_1 \in [1-\mu\delta,1],p_2 \in [0,1],p_1+p_2=1\}$. These parameter bounds are obtained by maximizing and minimizing $p_1$ and $p_2$ over $\mathcal{X}$. To ensure that Assumption~\ref{ass:positive_param} holds, i.e., $p\geq \0$ for all $p \in \mathcal{P}$, we select $\delta$ such that $1-\mu \delta \geq 0$.  Since the parameter depends on the current state, we synthesize a parameter independent RCI set by enforcing $Y=0$, $\munderbar{y}=y^0$ and $\munderbar{Y}=\0$. However, we allow $U^k$ to be freely computed, such that the invariance-inducing vertex control law is parameter dependent. Finally, we select $\mathcal{R}=\mathcal{P}$ for simplicity.
In Figure~\ref{figure:example_2}, we plot the sets computed by~\eqref{eq:final_SDP}, and compare the results with those presented in~\cite{Gupta2023}. The set $\mathcal{S}(p|y^0,\0)$ is parameterized with the following two choices of $C$: $(i)$  The same normal vectors found in the solution from~\cite{Gupta2023}; $(ii)$  The normal vectors of a $30$-sided uniform polytope, constructed as in~\cite[Remark 3]{villanueva2022configurationconstrained}. In both cases, we select $D$ to be a $30$-sided uniform polytope.
While the distance metric of the set computed in~\cite{Gupta2023} is $d_{\mathcal{X}}(\mathcal{S})=20.95$, our first parameterization results in $d_{\mathcal{X}}(\mathcal{S})=19.97$, and the second one in $d_{\mathcal{X}}(\mathcal{S})=18.15$, with smaller values indicating greater coverage of $\mathcal{X}$. These results suggest that the use of configuration-constrained polytopes  with parameter-dependent vertex control laws can result in RCI sets with reduced conservativeness for quasi-LPV systems.

\begin{figure}[t]
\centering
%	\vspace*{-0.7cm}
\hspace*{-0.0cm}
{
\resizebox{0.4\textwidth}{!}
{
\begin{tikzpicture}
\begin{scope}[xshift=0cm]
\node[draw=none,fill=none](tanks_fig) {\includegraphics[trim=0 0 0 0,clip,scale=2]{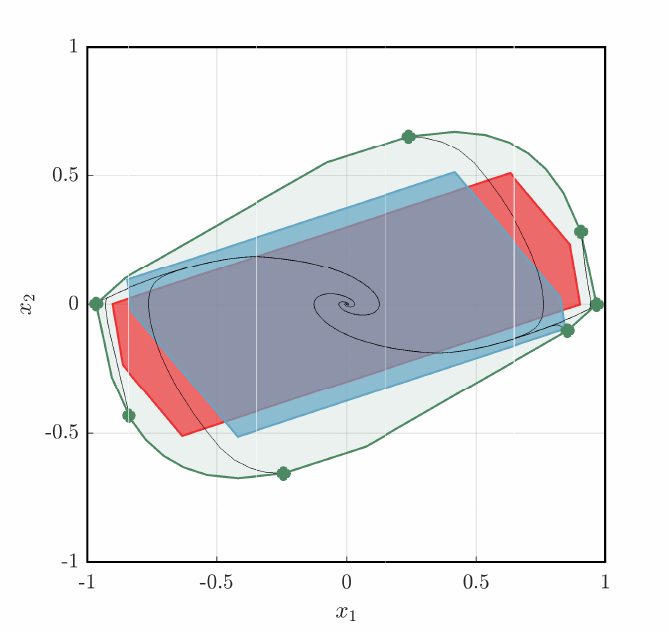}};
\end{scope}
\end{tikzpicture}}
}
\captionsetup{width=1\linewidth}
\caption{Results for Example~\ref{sec:Example2}. The green set denotes the RCI set we compute with rows of matrix $C$ representing the normal vectors of a $30$-sided uniform polytope. The blue set is the RCI set obtained using the approach of~\cite{Gupta2023}, and the red set in the RCI set obtained using our approach, with matrix $C$ chosen to be the same as the blue set. Closed-loop trajectories obtained using the vertex feedback law are plotted, illustrating invariance of the green set.}
\label{figure:example_2}
\end{figure}

\subsection{Vehicle lateral dynamics}
\begin{figure}[t]
\centering
\resizebox{0.42\textwidth}{!}
{
\begin{tikzpicture}
    % Node for the top part of the figure (bottom cropped out)
    \node[draw=none, fill=none] at (0,0) {\includegraphics[trim=0 12cm 0 1.5cm, clip, scale=1]{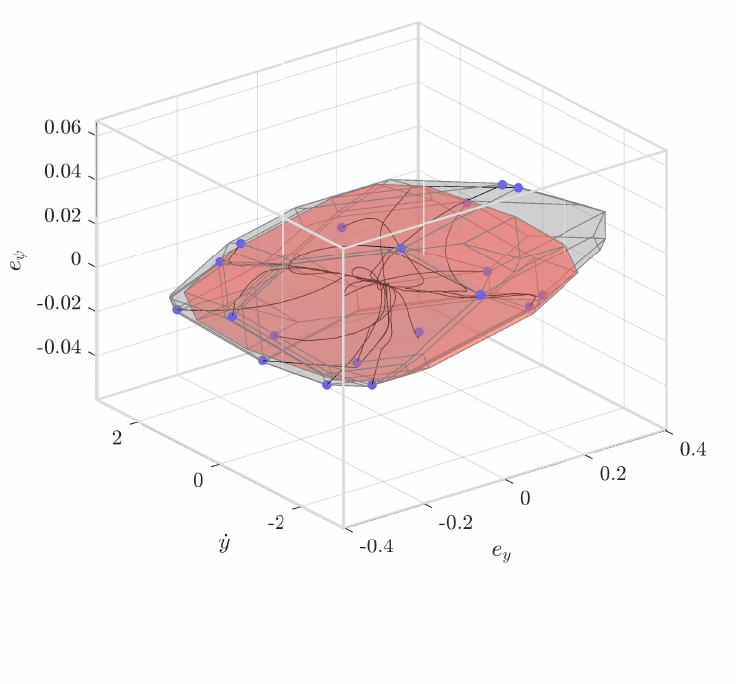}}; % Adjust the 5cm as needed to crop the bottom part
    
    % Node for the bottom part of the figure (top cropped out)
    \node[draw=none, fill=none] at (0,-9.) {\includegraphics[trim=0cm 1.9cm 0 13.7cm, clip, scale=1]{figure_1_R3_top.pdf}}; % Adjust the 5cm as needed to crop the top part
\end{tikzpicture}
}
\captionsetup{width=1\linewidth}
\caption{Projections of the sets $\tilde{\mathcal{S}}$ in grey, and $\mathcal{S}_{\mathrm{PI}}(\bm{W})$ in red. (\textit{Top}: Projection to $e_y$-$\dot{y}$-$e_{\psi}$ space, \textit{Bottom}: Projection to $e_y$-$\dot{y}$-$\dot{\psi}$ space.) Blue dots indicate $x(0)$ for several closed-loop trajectories, shown in black, resulting from the parameter-dependent vertex control law. The scheduling parameter sequences satisfy~\eqref{eq:bounded_param_variation} along with $p_2=1/p_1$, and the disturbance sequences are randomly sampled from $\mathcal{W}$. }
\label{fig:bicycle_solution}
\end{figure}

\label{sec:Example3}
We consider the problem of designing a PD-RCI set for vehicle lateral dynamics described by the bicycle model 
\begin{align}
\label{eq:lateral_dynamics}
     \dot{x} = (\mathbf{A}^0+v_x \mathbf{A}^1+ ({1}/{v_x}) \mathbf{A}^2) x+\mathbf{B}u+\mathbf{B}_{\mathrm{w}}w,
\end{align}
with state $x:=[e_y \ \dot{y} \ e_{\psi} \ \dot{\psi}]^{\top}$, where $e_y$ [rad] is the lateral error, $\dot{y}$ [m/s] the lateral velocity, $e_{\psi}$ [rad] the orientation error and $\dot{\psi}$ [rad/s] is the yaw rate. The input $u=[\delta_\mathrm{s} \ \mu_{\mathrm{b}}]^{\top}$, where $\delta_{\mathrm{s}}$ [rad] is the steering angle, and $\mu_{\mathrm{b}}$ [Nm] is the braking yaw moment. The disturbance $w = v_{\mathrm{w}}^2$, where $v_{\mathrm{w}} \in [-10,10]$ [m/s] is the wind velocity, and $v_x$ [m/s] is the current measured vehicle longitudinal velocity.
The model matrices from~\cite{ag19b} are given as
 \begin{align*}
 \scalemath{0.9}{
      \begin{matrix}
      &\overbrace{\begin{bmatrix} 0 & 1 & 0 & 0 \\ 0 & 0 & 0 & 0 \\ 0 & 0 & 0 & 1 \\ 0 & 0 & 0 & 0 \end{bmatrix}}^{\mathbf{A}^0} \ \   
      \overbrace{\begin{bmatrix} 0 & 0 & 1 & 0 \\ 0 & 0 & 0 & -1 \\ 0 & 0 & 0 & 0 \\ 0 & 0 & 0 & 0 \end{bmatrix}}^{\mathbf{A}^1} \ \ 
      \overbrace{\begin{bmatrix} 0 & 0 & 0 & 0 \\ 0 & -171.29 & 0 & 85.25 \\ 0 & 0 & 0 & 0 \\ 0 & 42.19 & 0 & -199.65 \end{bmatrix}}^{\mathbf{A}^2}, \vspace{5pt} \\
      &\underbrace{\begin{bmatrix} 0 & 65.8919 & 0 & 43.6411 \\ 0 & 0 & 0 & 0.2287e^{-3} \end{bmatrix}}_{\mathbf{B}^{\top}} \ \  \underbrace{\begin{bmatrix} 0 & 0.0018 & 0 &-0.0022 \end{bmatrix}}_{\mathbf{B}_{\mathrm{w}}^{\top}}.
      \end{matrix}
      }
\end{align*}
The states are constrained as $\mathcal{X}=\{x: |x| \leq [0.4 \ 3 \ 10\pi/180 \ 1]^{\top}\}$ and inputs as $\mathcal{U}=\{u: |u| \leq [2.5 \pi/180 \ 1]^{\top}\}$.
To design a PD-RCI set for this system, we discretize the dynamics using the forward Euler scheme with timestep of $0.025$s, and consider the parameter vector $p=[v_x \ 1/v_x \ 1]^{\top}$. Then, we obtain the system matrices in~\eqref{eq:P_dependence} as $A^1=\delta \mathbf{A}^1$, $A^2=\delta \mathbf{A}^2$, ${A^3=\delta \mathbf{A}^0+\I^4, B^1,B^2=\0, B^3=\delta \mathbf{B}}$, and the disturbance set as $\mathcal{W}=\{\delta \mathbf{B}_w w : w \in [0,100]\}$. We assume that the longitudinal velocity is bounded as $v_x \in [20,100]/3.6$ [m/s], where division by $3.6$ is performed to convert [Km/h] to [m/s].
Based on these bounds, we define the parameter set as $\mathcal{P}=\hat{\mathcal{P}} \times \{1\}$, where 
\begin{align*}
%\label{eq:graph_approximation}
\scalemath{1}{
    \hat{\mathcal{P}}:=\mathrm{CH}\left\{\begin{bmatrix} 5.556 \\ 0.18 \end{bmatrix}, \begin{bmatrix} 27.777 \\ 0.036 \end{bmatrix}, \begin{bmatrix} 19.289 \\ 0.036 \end{bmatrix}, \begin{bmatrix} 5.556 \\ 0.125 \end{bmatrix}\right\}}
\end{align*}
overapproximates the set $\{p : p_1 \in [20,100]/3.6, \ p_2=1/p_1\}$.
%The polytope $\hat{\mathcal{P}}$ overapproximates the true parameter set, as shown in Figure~\ref{figure:bicycle_parameter}. 
To define the increment set $\mathcal{R}$, we assume $|v_x^+-v_x| \leq 1$, or equivalently $|p_1^+-p_1|\leq 1$.
We then derive bounds on the variation of parameter $p_2=1/p_1$ as follows. For any $p_1\in [20,100]/3.6$, $p_1^+ \in [p_1-1,p_1+1] \cap [20,100]/3.6$ holds. Then, the variation $\tilde{p}_2=p_2^+-p_2=(p_1 - p_1^+)/(p_1p_1^+)$
is maximized if we have $p_1^+=\max(p_1-1,20/3.6)$, and minimized if $p_1^+=\min(p_1+1,100/3.6)$. Thus, the bounds on $\tilde{p}_2$ depend nonlinearly on $p_1$. For simplicity, we select the largest and smallest values of these bounds over all $p_1 \in [20,100]/3.6$ as bounds on $\tilde{p}_2$, i.e., $p_1=(20/3.6)+1$ and $p_1^+=20/3.6$ result in the largest value of $\tilde{p}_2=0.02746$, and $p_1=20/3.6$ and $p_1^+=(20/3.6)+1$ result in the smallest value of $\tilde{p}_2=-0.02746$. Thus, $\mathcal{R}=\{\tilde{p} : |\tilde{p}| \leq [1 \ 0.02746 \ 0]^{\top}\}$ describes the increment set.
\begin{remark}
    While we model $p_1$ and $p_2$ as independent parameters verifying $(p_1,p_2) \in \hat{\mathcal{P}}$, they satisfy $p_1p_2=1$ in reality. Future work can focus on exploiting this dependence to synthesize PD-RCI sets with conservativeness further reduced. $\hfill\square$
\end{remark}
%\begin{figure}[t]
%\centering
%\hspace*{-0.2cm}
%{
%\resizebox{0.5\textwidth}{!}
%{
%\begin{tikzpicture}
%\begin{scope}[xshift=0cm]
%\node[draw=none,fill=none](tanks_fig) {\includegraphics[trim=0 0 0 0,clip,scale=2.]%{bicycle_parameter_set.pdf}};
%\end{scope}
%\end{tikzpicture}}
%}
%\captionsetup{width=1\linewidth}
%\caption{The set $\hat{\mathcal{P}}$ in~\eqref{eq:graph_approximation} is designed to overapproximate %the true parameter set $\{p:p_1p_2=1,p_1 \in [20,100]/3.6\}$. Conservativeness can be reduced by %refining $\hat{\mathcal{P}}$ further.}
%\label{figure:bicycle_parameter}
%\end{figure}
%\begin{remark}
%The PD-RCI set computation method can directly accomodate the parameter set $\scalemath{0.95}{\mathcal{P}=\{p:p_1p_2=1,p_1 \in [20,100]/3.6\} \times 1}$ modeled using QCs as per Remark~\ref{remark:QC}. Additionally, parameter-dependent constraints can be accommodated. For instance, constraints like limiting the lateral velocity at high longitudinal speeds for passenger comfort can be captured by substituting the state constraint set $\mathcal{X}$ with $\mathcal{X}(p)=\{x:H^x x \leq h^x - \hat{\kappa} p_1\}$, where $\hat{\kappa} \geq \0$ denotes the desired tightening. The inequalities in~\eqref{eq:state_incl_1} are then encoded as per Proposition~\ref{prop:strong_duality}. Implementing and analyzing such extensions will be addressed in future research. $\hfill\square$
%$\hfill\square$
%\end{remark}

To derive the matrix $C$ that parameterizes the PD-RCI set in~\eqref{eq:PD_RCI_parameterization}, we implement the procedure detailed in Section~\ref{sec:selecting_C}. We choose $\hat{C}$ in~\eqref{eq:gupta_vertex_set} with $m_s=24$ normal vectors as $\hat{C}=[\I^4 \ -\I^4 \ 0.75\bm{C}]^{\top}$, where $\bm{C} \in \R^{4 \times 16}$ represents the $16$ vertices of the set $\mathcal{B}_{\infty}^{4}$ arranged column-wise. This selection leads to $N=48$ vertices for the set $\mathcal{S}_{\mathrm{PI}}(\bm{W})$. We select matrix $D=\bm{C}^{\top}$ to formulate the objective of Problem~\eqref{eq:NLP_init_C}.
Solving Problem~\eqref{eq:NLP_init_C} takes $8.6592$ s, yielding 
$$
\bm{W}=\begin{bmatrix}
    0.3819  & -0.0432  & -0.0542  &  0.0438 \\
    0.0057  &  2.8432  & -0.1253  &  0.4704 \\
   -0.0225  & -0.0423  &  0.0241  & -0.0451 \\
    0.0069  &  0.0712  &  0.0583  &  0.6544
\end{bmatrix}.
$$
Then we select $C=\hat{C}\bm{W}^{-1}$ to parameterize the PD-RCI set $\mathcal{S}(p|y^0,Y)$. Using this parameterization, we formulate and solve Problem~\eqref{eq:final_SDP}. The problem building and solution time amounts to $8.4223$ s. The sets obtained are plotted in Figure~\ref{fig:bicycle_solution}. We also plot several simulated trajectories of the plant, with control input computed as in~\eqref{eq:cc_control_1}. The computed PD-vertex control law successfully regulates the system from all $x(0) \in \mathcal{S}(p(0)|y^0,Y)$.  
We report that the volume of $\mathcal{S}_{\mathrm{PI}}(\bm{W})$ is $0.0327$, while the volume of the set $\tilde{\mathcal{S}}$ is $0.0459$. This highlights that we obtain a larger region of attraction by considering PD-RCI sets. 

% \begin{remark}
% The PD-RCI set computation method can directly accomodate the parameter set $\scalemath{0.95}{\mathcal{P}=\{p:p_1p_2=1,p_1 \in [20,100]/3.6\} \times 1}$ modeled using QCs as per Remark~\ref{remark:QC}. Additionally, parameter-dependent constraints can be accommodated. For instance, constraints like limiting lateral velocity at high longitudinal speeds for passenger comfort can be captured by substituting the state constraint set $\mathcal{X}$ with $\mathcal{X}(p)=\{x:H^x x \leq h^x - \hat{\kappa} p_1\}$, where $\hat{\kappa} \geq \0$ signifies desired tightening. The inequalities in~\eqref{eq:state_incl_1} are then encoded as per Proposition~\ref{prop:strong_duality}. Implementing and analyzing such extensions will be addressed in future research.
% $\hfill\square$
% \end{remark}
%\begin{figure}[t]
%\centering
%	\vspace*{-0.7cm}
%\hspace*{-0.6cm}
%{
%\resizebox{0.5\textwidth}{!}
%{
%\begin{tikzpicture}
%\begin{scope}[xshift=0cm]
%\node[draw=none,fill=none](tanks_fig) {\includegraphics[trim=0 0 0 0,clip,scale=2.5]{bicycle_plots_final.pdf}};
%\end{scope}
%%\end{tikzpicture}}
%}
%\captionsetup{width=1\linewidth}
%\caption{Projections of the sets $\tilde{\mathcal{S}}$ and $\mathcal{S}_{\mathrm{PI}}(\bm{W})$. Several trajectories of the system are plotted, with green dots indicating the initial state $x(0) \in \mathcal{S}(p(0)|y^0,Y)$ for some $p(0) \in \mathcal{P}$.}
%\label{figure:bicycle_solution}
%\end{figure}

\section{Conclusions}
We have presented a method to synthesize PD-RCI sets for LPV systems, with invariance induced using PD-vertex control laws. These sets and control laws are computed as the solution of a single SDP problem, formulated by exploiting properties of configuration-constrained polytopes. The method outperforms state-of-the-art approaches, both with respect to conservativeness and computational burden. Owing to the fact that the set of PD-RCI sets now has a convex characterization, future work aims at synthesizing tube-based model predictive control schemes using these sets for LPV systems with bounded parameter variation. Moreover, the study of data-driven characterization of such sets, and extending the methods to rational parameter dependence are current research directions.

 \bibliography{references}

 \section{Appendix : Construction of configuration-constraints}
 \label{sec:appendix}
Given a matrix $C$, we recall a method from~\cite[Section 3.5]{villanueva2022configurationconstrained} to construct matrix $\mathbf{E}$ that characterizes a vertex configuration domain of the polytope $\mathcal{Y}(y):=\{x:Cx \leq y\}$, such that $\mathbf{E}y \leq 0$ implies $\mathcal{Y}(y)=\mathrm{CH}\{V^{k}y, k \in \mathbb{I}_1^N\}$. Observe that this implication is equivalent to~\eqref{eq:config_constraint_relation}. Here, we focus on a particular methodology for completeness. For further technical results, the reader is referred to~\cite[Theorem 2]{villanueva2022configurationconstrained}.

Suppose that for some $\sigma \in \R^{m}$, the polytope $\mathcal{Y}(\sigma)$ has $N$ unique vertices, and each vertex results from the intersection of exactly $n$ hyperplanes (minimal representation). Denote the indices of the hyperplanes intersection at vertex $k$ as $J_k$, with $|J_k|=n$ and $\max\{J_k\}\leq m_s$ for each $k \in \mathbb{I}_1^N$. Let $C_{J_{k}} \in \R^{n \times n}$ (corres. $\sigma_{J_{k}}$) be a matrix constructed using rows $J_{k}$ of $C$ (corres. $\sigma$), such that vertex $k$ of $\mathcal{Y}(\sigma)$ is given by $C_{J_{k}}^{-1}\sigma_{J_{k}}$. Then, construct the matrices $V^k \in \R^{n \times m_s}$ with zeros everywhere except in columns $J_{k}$, with these columns populated by the columns of $C_{J_{k}}^{-1}$. It then follows that vertex $k$ of $\mathcal{Y}(\sigma)$ is $V^k \sigma$. Using these matrices $V^k$, the vertex configuration domain matrix $\mathbf{E}$ can be constructed as
$$
\mathbf{E}=\begin{bmatrix} CV^1-\I^{m_s} \\ \vdots \\ CV^N-\I^{m_s} \end{bmatrix}.
$$
This matrix, along with $V^k$ as guaranteed to satisfy~\eqref{eq:config_constraint_relation} as per~\cite[Theorem 2]{villanueva2022configurationconstrained}. The selection of $(C,\sigma)$ that is optimal for our application is a  topic for future research.

\end{document}